%% file: Journal.tex
\documentclass[conference]{IEEEtran}
\IEEEoverridecommandlockouts
\usepackage{cite}
\usepackage{amsmath,amssymb,amsfonts}
\usepackage{graphicx}
\usepackage{textcomp}
\usepackage{xcolor}
\usepackage{subfiles}
\usepackage[utf8]{inputenc}
\usepackage[english]{babel}
\usepackage{amsthm}
\usepackage{subcaption}
\usepackage{caption}
\usepackage{stackengine}
\usepackage{mathtools}
\usepackage{stmaryrd}
\usepackage{url}

\newtheorem{theorem}{Theorem}

\theoremstyle{definition}
\newtheorem{definition}{Definition}

\theoremstyle{fact}

\theoremstyle{remark}

\def\BibTeX{{\rm B\kern-.05em{\sc i\kern-.025em b}\kern-.08em
    T\kern-.1667em\lower.7ex\hbox{E}\kern-.125emX}}

\makeatletter
\def\endthebibliography{%
	\def\@noitemerr{\@latex@warning{Empty `thebibliography' environment}}%
	\endlist
}
\makeatother

\begin{document}

\subfile{subfiles/authors}

\begin{abstract}

	This work presents a decentralized allocation algorithm of safety-critical application on parallel computing architectures, where individual Computational Units can be affected by faults. 
	
	The described method consists in representing the architecture by an abstract graph where each node represents a Computational Unit. Applications are also represented by the graph of Computational Units they require for execution. The problem is then to decide how to allocate Computational Units to applications to guarantee execution of the safety-critical application. The problem is formulated as an optimization problem, with the form of an Integer Linear Program. A state-of-the-art solver is then used to solve the problem.
	
	Decentralizing the allocation process is achieved through redundancy of the allocator executed on the architecture. No centralized element decides on the allocation of the entire architecture, thus improving the reliability of the system. 
	
	Experimental reproduction of a multi-core architecture is also presented. It is used to demonstrate the capabilities of the proposed allocation process to maintain the operation of a physical system in a decentralized way while individual component fails. 

\end{abstract}


\section{Introduction and prior art}\label{sec:introduction}

\subfile{subfiles/introduction}

\section{Theoretical aspect}
\subsection{Mathematical description of the allocation problem}

\subfile{subfiles/math_formulation}

\section{ILP formulation of the task allocation problem}\label{sec:algo}

\subfile{subfiles/allocation_algo}

\section{Decentralization of the allocation system}\label{sec:decentr}

\subfile{subfiles/decentralization}

\section{Practical example}

\input{subfiles/practical_expl.tex}


\section{Conclusion}\label{sec:conclusion}

\subfile{subfiles/conclusion}



\section*{Appendix}\label{sec:appendix}

\input{subfiles/appendix_v2.tex}

\bibliographystyle{IEEEtran}
\bibliography{references}

\end{document}


\section*{Appendix 2: NoC model}
\begin{itemize}
	
	\item
	The matrix $G$ is the $N_{CR} \times N_{paths}$ incidence matrix of the graph representing the NoC, where $N_{CR}$ is the number of CRs on the fabric and $N_{paths}$ is the number of communication paths.
	\[G_{ij}=\begin{cases} 
	1 & \text{if the CR $i$ and the NoC path $j$ are incident} \\
	0 & \text{otherwise} 
	\end{cases}\]
	
	\item
	The matrix $A^{k}$  is the $N^{k}_{nodes} \times N^{k}_{links}$ incidence matrix representing the application $k$,  where $N^{k}_{nodes}$ is its number of nodes and $N^{k}_{links}$ its number of links.
	\[A^{k}_{ij}=\begin{cases} 
	1 & \parbox[t]{.75\linewidth}{if the node $i$ and the link $j$ of application $k$ are incident} \\
	0 & \text{otherwise}
	\end{cases}\]

\end{itemize}

$N_{nodes} = \sum_{k=1}^{N_{apps}}  N^{k}_{nodes}$ is the total number of application nodes.

$N_{links} = \sum_{k=1}^{N_{apps}}  N^{k}_{links}$ is the total number of application links.\\

\begin{itemize}
	
	\item
	$T^{CR}$ is a $N_{CR} \times N^{CR}_{types}$ incidence matrix describing the types of each Compute Resource. $N^{CR}_{types}$ is the number of different CR types.
	
	\[T^{CR}_{ij}=\begin{cases} 
	1 & \text{if the CR $i$ is of type $j$} \\
	0 & \text{otherwise} 
	\end{cases}\]
	
	\item
	$T^{paths}$ is a $N_{paths} \times N^{paths}_{types}$ incidence matrix describing the types of each NoC path. $N^{paths}_{types}$ is the number of different communication paths types.
	
	\[T^{paths}_{ij}=\begin{cases} 
	1 & \text{if the NoC path $i$ is of type $j$} \\
	0 & \text{otherwise} 
	\end{cases}\]
	
	\item
	For each application $A^{k}$, the type of each node is described in the $N^{k}_{nodes} \times N^{CR}_{types}$ incidence matrix $TN^{k}$. If a given node is of a certain type, it will need to be allocated to a CR of the same type.
	\[TN^{k}_{ij}=\begin{cases} 
	1 & \text{if the node $i$ is of type $j$ in application $k$} \\
	0 & \text{otherwise} 
	\end{cases}\]
	
	\item
	For each application $A^{k}$, the type of each link is described in the $N^{k}_{links} \times N^{paths}_{types}$ incidence matrix $TL^{k}$.
	\[TL^{k}_{ij}=\begin{cases} 
	1 & \text{if the link $i$ is of type $j$ in application $k$} \\
	0 & \text{otherwise} 
	\end{cases}\]
	
\end{itemize}

\begin{itemize}
	
	\item
	An overall application graph represented by the incidence matrix $A$ is constructed from the $A^{k}$ $\left(k=1, \dotsc ,N_{apps}\right)$ matrices as such:
	\[
	A = \begin{bmatrix} 
	A^{1} &   &   \\
	& \ddots &   \\
	&   & A^{k} 
	\end{bmatrix}
	\quad \text{is a $N_{nodes} \times N_{links}$ matrix.}
	\]
	
	The same is done with the $TN^{k}$ and $TL^{k}$ $\left(k=1, \dotsc ,N_{apps}\right)$ matrices:
	\[
	TN = \begin{bmatrix} 
	TN^{1}\\
	\vdots \\
	TN^{k} 
	\end{bmatrix}
	\quad \text{is a $N_{nodes} \times N^{CR}_{types}$ matrix.}
	\]
	
	\[
	TL = \begin{bmatrix} 
	TL^{1}\\
	\vdots \\
	TL^{k} 
	\end{bmatrix}
	\quad \text{is a $N_{links} \times N^{paths}_{types}$ matrix.}
	\]
\end{itemize}

$ $
\\


This section provides a proof of the statement in \ref{subsec:opt_prob}. Prior to the Theorem, some useful statements are given as follows

\begin{definition}
	The optimal solution $\boldsymbol{x}^*$ to an optimization problem described in Sec. \ref{subsubsec:obj_func} is a vectorization of the optimal value of all decision variables described in Sec. \ref{subsubsec:decision_vars}.
\end{definition}


\begin{fact}\label{fact:Nnodes}
	$\forall j \leq N_{\textrm{apps}}$, $N_{\textrm{nodes}} \geq \sum_{i=j}^{N_{\textrm{apps}}} N^{i}_{\textrm{nodes}} \geq N^j_{\textrm{nodes}} > 0 . $
\end{fact}

\begin{lemma}\label{lemma:alpha}
	Let $\alpha_k$ be given by the Eq. \ref{eq:obj_func_coeff}, then 
$\forall i < N_{\textrm{apps}},\  \forall j \leq N_{\textrm{apps}},\ i < j\ ;\ \alpha_i > \sum_{k = j}^{N_{\textrm{apps}}}\alpha_k > 0 . $
\end{lemma}

\begin{proof}
	\begin{align*}
	\alpha_i &= \sum_{l = i+1}^{N_{\textrm{apps}}} \alpha_l +  (\beta+1) \times N_{\textrm{nodes}} + 1 > 0\\
	\alpha_i &> \sum_{l = i+1}^{N_{\textrm{apps}}} \alpha_l \\
	&\geq \sum_{k = j}^{N_{\textrm{apps}}}\alpha_k
	\end{align*}
	since $i < j \Rightarrow i + 1 \leq j$. Q.E.D.
\end{proof}

Now, we are ready to state the Theorem.

\begin{theorem}\label{theorem:alpha}
	Given the optimization problem described in Sec. \ref{subsubsec:obj_func}, and if by reallocating the non-lowest priority application $i$ requires other $N_{i'}$ applications to be reallocated as well as another $N_{i''}$ lower priority applications to be dropped, then its optimal solution $\boldsymbol{x}^*$ always reallocates $N_{i'}$ applications and drops $N_{i''}$ lower priority applications in order to reallocate an application $i$ if there are not enough CRs to execute that application unless some lower priority ones are dropped.
\end{theorem}

\begin{proof}
	By contradiction, suppose there exists an optimal solution $\hat{\boldsymbol{x}}^*$ that drops the non-lowest priority application $i$ in order to keep other $N_{i'}$ applications and $N_{i''}$ lower priority applications, which means $f(\hat{\boldsymbol{x}}^*) > f(\boldsymbol{x}^*)$. The objective function is
	
	\begin{align}
	\sum_{k = 1}^{N_{\textrm{apps}} } \alpha_k \cdot r_k - (\beta+1)\sum_{j = 1}^{N_{\textrm{nodes}}} M_j - X^{\textrm{Comm}},
	\end{align}
	where 
	\begin{align}
	\begin{split}
	X^{\textrm{Comm}} &= \sum_{k = 1}^{N_{\textrm{realloc}}}\sum_{j = 1}^{N_{\textrm{CRs}}}\sum_{i = 1}^{N_{\textrm{paths}}} \left|X^{\textrm{Comm, } k}_{ij}\right| , \\
	\beta &= N_{\textrm{realloc}} \times N_{\textrm{CRs}} \times N_{\textrm{paths}} \\
	\alpha_k &= \sum_{l = k+1}^{N_{\textrm{apps}}} \alpha_l + (\beta+1) \times N_{\textrm{nodes}} + 1 .
	\end{split}
	\end{align}
	

	Let $\mathcal{I}' = \{i'_{1} \ldots i'_{N_{i'}} \}$ be the set of the $N_{i'}$ applications required to be allocated, $\mathcal{I}'' = \{i''_{1} \ldots i''_{N_{i''}} \}$ be the set of the other $N_{i''}$ applications required to be dropped and $\mathcal{I} = \{i\} \cup \mathcal{I}' \cup \mathcal{I}''$, and by dropping an application $i$ means that the value of $r_k$ corresponding to an application $i$ must be zero as well as every element of a decision vector $M$ since no Application Nodes are required to moved. Therefore, the objective function for $\hat{\boldsymbol{x}}^*$ is
	\begin{align}
	\begin{split}	
	f(\hat{\boldsymbol{x}}^*) &= \sum_{k = 1|k \notin \mathcal{I}}^{N_{\textrm{apps}}} \alpha_k \cdot r_k + \sum_{k \in \mathcal{I}' \cup \mathcal{I}''}^{} \alpha_k \cdot 1 - X^{\textrm{Comm}}.
	\end{split}
	\end{align}
	
	On the other hand, for the objective function $f(\boldsymbol{x}^*)$, the value of $r_k$ and $M_j$(s) corresponding to an application $i$ must be one, since it has been reallocated in order to keep it running as well as the ones corresponding the other $i'$ applications. However, the value of $r_k$ and $M_j$ corresponding the other $N_{i'}$ applications must be zero since they have been dropped. Therefore, the objective function for $\boldsymbol{x}^*$ is
	
	\begin{align}
	\begin{split}
	f(\boldsymbol{x}^*) = &\sum_{k = 1|k \notin \mathcal{I}}^{N_{\textrm{apps}}} \alpha_k \cdot r_k + \sum_{k \in \{i\} \cup \mathcal{I}'}^{} \alpha_k \cdot 1\\
	&- (\beta+1)\sum_{l \in \{i\} \cup \mathcal{I}'}^{} \left(\sum_{j = 1}^{N^{l}_{\textrm{nodes}}} 1\right) - X^{\textrm{Comm}}.
	\end{split}
	\end{align}
	Then, 
	\begin{align}
	\begin{split}
	f(\hat{\boldsymbol{x}}^*) - f(\boldsymbol{x}^*) = &\sum_{k \in \mathcal{I}''}^{} \alpha_k - \left(\alpha_i - (\beta+1)\sum_{l \in \{i\} \cup \mathcal{I}'}^{} N^{l}_{\textrm{nodes}}\right)\\	
	= &\sum_{k \in \mathcal{I}''}^{} \alpha_k + (\beta+1)\sum_{l \in \{i\} \cup \mathcal{I}'}^{} N^{l}_{\textrm{nodes}}\\
	& - \left(\sum_{l = i+1}^{N_{\textrm{apps}}} \alpha_l + (\beta+1) \times N_{\textrm{nodes}} + 1\right)\\
	= &\left(\sum_{k \in \mathcal{I}''}^{} \alpha_k - \sum_{l = i+1}^{N_{\textrm{apps}}} \alpha_l \right) - 1\\
	&+ (\beta+1)\left(\sum_{l \in \{i\} \cup \mathcal{I}'}^{} N^{l}_{\textrm{nodes}} - N_{\textrm{nodes}}\right).
	\end{split}
	\end{align}
	By Fact \ref{fact:Nnodes} and Lemma \ref{lemma:alpha},
	\begin{align}
	\begin{split}
	f(\hat{\boldsymbol{x}}^*) - f(\boldsymbol{x}^*) = &- \sum_{k>i|k \notin \mathcal{I}''}^{} \alpha_k - (\beta+1)\sum_{l \notin \{i\} \cup \mathcal{I}'}^{} N^{l}_{\textrm{nodes}} - 1\\
	f(\hat{\boldsymbol{x}}^*) - f(\boldsymbol{x}^*) <&\ 0\\
	f(\hat{\boldsymbol{x}}^*) <&\ f(\boldsymbol{x}^*)
	\end{split}
	\end{align}
	which is a contradiction. Thus, the statement in Sec. \ref{subsec:opt_prob} is proven. Q.E.D.
\end{proof}


This section provides a proof of the statement in \ref{subsec:opt_prob}. Prior to the Theorem, some useful statements are given as follows

\begin{definition}
	The optimal solution $\mathbf{x}^*$ to an optimization problem described in Sec. \ref{subsubsec:obj_func} is a vectorization of the optimal value of all decision variables described in Sec. \ref{subsubsec:decision_vars}.
\end{definition}

\begin{definition}
	For each application $i \in \llbracket 1, N_{apps} \rrbracket$, we define the set $\mathcal{N}^{\mathcal{I}}$ as the set of the Application Nodes of application $i$.
	Furthermore, $\mathcal{N}^{\mathcal{I}}$ is a set of Application Nodes required to run a set, $\mathcal{I} = \{i_1 \ldots i_I\}$, of $I$ applications.
\end{definition}

\begin{fact}\label{fact:Nnodes}
	$\forall j \leq N_{apps}$, $N_{nodes} \geq \sum_{i=j}^{N_{apps}} N^{i}_{nodes} \geq N^j_{nodes} > 0 $
\end{fact}

\begin{fact}\label{fact:alpha}
	Let $\alpha_{N_{apps}} \in \mathbb{N}$ and $\forall i < N_{apps}$, $\alpha_i = \sum_{k = i+1}^{N_{apps}} \alpha_k$, then 
	$\forall i < j \leq N_{apps}$, $\alpha_i \geq \sum_{k=j}^{N_{apps}}\alpha_k > 0$
\end{fact}

\begin{lemma}\label{lemma:alpha_bis}
	Let $\alpha_k$ be given by the Eq. \ref{eq:obj_func_coeff}, then 
$\alpha_i > \sum_{k = j}^{N_{apps}}\alpha_k\ ;\ \forall i < N_{apps},\  \forall j \leq N_{apps},\ i < j$
\end{lemma}

\begin{proof}
	The proof of the inequality part of the statement follows from Fact \ref{fact:alpha_bis}. For the strictly inequality part, consider the base case where $i = N_{apps} - 1$ and $j = N_{apps}$, 
	\begin{align}
	\begin{split}	
	\alpha_{N_{apps} - 1} &> \sum_{k = N_{apps}}^{N_{apps}}\alpha_{k}\\
	\alpha_{N_{apps}} + \sum_{l = N_{apps} - 1}^{N_{apps}} N^l_{nodes} + N_{nodes} + 1	&> \alpha_{N_{apps}}
	\end{split}
	\end{align}
	which is true. Also, suppose that the case where $i + 1 < N_{apps}$ and $j + 1 \leq N_{apps}$, i.e. $\alpha_{i + 1} > \sum_{k=j+1}^{N_{apps}}\alpha_k$,
	is true. By induction, we will show that the case where $i < N_{apps}$ and $j \leq N_{apps}$, i.e. $\alpha_{i} > \sum_{k=j}^{N_{apps}}\alpha_{k}$, is true. 
	By canceling the term $N_{nodes} + 1$ on both sides, 
	\begin{align}
	\begin{split}
	\sum_{l = i+1}^{N_{apps}} \alpha_l + \sum_{l = i}^{N_{apps}} N^l_{nodes} >
	\sum_{k = j+1}^{N_{apps}} \alpha_k &+ \sum_{l = j}^{N_{apps}} N^l_{nodes}\\
	\alpha_{i+1} + \sum_{l = i+2}^{N_{apps}} \alpha_l + \sum_{l = i}^{N_{apps}} N^l_{nodes} >
	\sum_{k = j+1}^{N_{apps}} \alpha_k &+ \sum_{l = j}^{N_{apps}} N^l_{nodes}\\
	\left(\alpha_{i+1} - \sum_{k = j+1}^{N_{apps}} \alpha_k\right) + \sum_{l = i+2}^{N_{apps}} \alpha_l\ &+ \\
\left(\sum_{l = i}^{N_{apps}} N^l_{nodes} - \sum_{l = j}^{N_{apps}} N^l_{nodes}\right) &> 0\\
	\end{split}
	\end{align}
	which is true by the induction hypothesis, Fact \ref{fact:Nnodes_bis}, and Fact \ref{fact:alpha_bis}. Q.E.D.
\end{proof}

Now, we are ready to state the Theorem.

\begin{theorem}\label{theorem:alpha_bis}
	Given the optimization problem described in Sec. \ref{subsubsec:obj_func}, and if by reallocating the non-lowest priority application $i$ requires other $N_{i'}$ applications to be reallocated as well as another $N_{i''}$ lower priority applications to be dropped, then its optimal solution $\mathbf{x}^*$ always reallocates $N_{i'}$ applications and drops $N_{i''}$ lower priority applications in order to reallocate an application $i$ if there are not enough CRs to execute that application unless some lower priority ones are dropped.
\end{theorem}

\begin{proof}
	By contradiction, suppose there exists an optimal solution $\hat{\mathbf{x}}^*$ that drops the non-lowest priority application $i$ in order to keep other $N_{i'}$ applications and $N_{i''}$ lower priority applications, which means $f(\hat{\mathbf{x}}^*) > f(\mathbf{x}^*)$. Since the objective function
	\begin{align}
	\begin{split}	
	f(\mathbf{x}) &= \sum_{k = 1}^{N_{apps}} \alpha_k \cdot r_k - \sum_{j = 1}^{N_{nodes}} M_j
	\end{split}
	\end{align}
	where
	\begin{align}
	\begin{split}
	\alpha_k &= \sum_{l = k+1}^{N_{apps}} \alpha_l + \sum_{l = k}^{N_{apps}} N^l_{nodes} + N_{nodes} + 1
	\end{split}
	\end{align}
	Let $\mathcal{I}' = \{i'_{1} \ldots i'_{N_{i'}} \}$ be the set of the $N_{i'}$ applications required to be allocated, $\mathcal{I}'' = \{i''_{1} \ldots i''_{N_{i''}} \}$ be the set of the other $N_{i''}$ applications required to be dropped and $\mathcal{I} = \{i\} \cup \mathcal{I}' \cup \mathcal{I}''$, and by dropping an application $i$ means that the value of $r_k$ corresponding to an application $i$ must be zero as well as every element of a decision vector $M$ since no Application Nodes are required to moved. Therefore, the objective function
	\begin{align}
	\begin{split}	
	f(\hat{\mathbf{x}}^*) &= \sum_{k = 1|k \notin \mathcal{I}}^{N_{apps}} \alpha_k \cdot r_k + \sum_{k \in \mathcal{I}' \cup \mathcal{I}''}^{} \alpha_k \cdot 1
	\end{split}
	\end{align}
	
	On the other hand, for the objective function $f(\mathbf{x}^*)$, the value of $r_k$ and $M_j$(s) corresponding to an application $i$ must be one, since it has been reallocated in order to keep it running as well as the ones corresponding the other $i'$ applications. However, the value of $r_k$ and $M_j$ corresponding the other $N_{i'}$ applications must b zero since they have been dropped. Therefore, the objective function
	
	\begin{align}
	\begin{split}
	f(\mathbf{x}^*) = &\sum_{k = 1|k \notin \mathcal{I}}^{N_{apps}} \alpha_k \cdot r_k + \sum_{k \in \mathcal{I}'}^{} \alpha_k \cdot 1 + \alpha_i \cdot 1\\
	&- \sum_{l \in \mathcal{I}'}^{} \left(\sum_{j = 1}^{N^{l}_{nodes}} 1\right) - \sum_{j = 1}^{N^i_{nodes}} 1
	\end{split}
	\end{align}
	Then, 
	\begin{align}
	\begin{split}
	f(\hat{\mathbf{x}}^*) - f(\mathbf{x}^*) = &\sum_{k \in \mathcal{I}''}^{} \alpha_k - \left(\alpha_i - \sum_{l \in \mathcal{I}'}^{} N^{l}_{nodes} - N^i_{nodes}\right)\\	
	= &\sum_{k \in \mathcal{I}''}^{} \alpha_k + \sum_{l \in \mathcal{I}'}^{} N^{l}_{nodes} + N^i_{nodes}\\
	& - \left(\sum_{l = i+1}^{N_{apps}} \alpha_l + \sum_{l = i}^{N_{apps}} N^l_{nodes} + N_{nodes} + 1\right)\\
	= &\left(\sum_{k \in \mathcal{I}''}^{} \alpha_k - \sum_{l = i+1}^{N_{apps}} \alpha_l \right) - 1\\
	&+ \left(\sum_{l \in \mathcal{I}'}^{} N^{l}_{nodes} - N_{nodes}\right)\\
	&+ \left(N^i_{nodes} - \sum_{l = i}^{N_{apps}} N^l_{nodes}\right)
	\end{split}
	\end{align}
	By Lemma \ref{lemma:alpha},
	\begin{align}
	\begin{split}
	f(\hat{\mathbf{x}}^*) - f(\mathbf{x}^*) = &- \sum_{k>i|k \notin \mathcal{I}''}^{} \alpha_k - 1\\
	&- \sum_{l \notin \mathcal{I}'}^{} N^{l}_{nodes} - \sum_{l = i+1}^{N_{apps}} N^l_{nodes} < 0\\
	f(\hat{\mathbf{x}}^*) &< f(\mathbf{x}^*)
	\end{split}
	\end{align}
	which contradicts. Thus, the statement in Sec. \ref{subsec:opt_prob} is proven. Q.E.D.
\end{proof}

%% file: subfiles/authors.tex
\title{Decentralized On-line Task Reallocation on Parallel Computing Architectures with Safety-Critical Applications\\ 
	\thanks{This effort has been funded in part by SAFRAN and by the National Science Foundation, Grants CNS 1544332 and 1446758.}
}



\author{\IEEEauthorblockN{Thanakorn Khamvilai}
	\IEEEauthorblockA{\textit{School of Aerospace Engineering} \\
		\textit{Georgia Institute of Technology}\\
		Atlanta, GA, USA \\
		tkhamvilai3@gatech.edu}\\
	\IEEEauthorblockN{Philippe Baufreton}
	\IEEEauthorblockA{
	\textit{Safran Electronics \& Defense}\\
	Massy, France \\
	philippe.baufreton@safrangroup.com}
	\and
	\IEEEauthorblockN{Louis Sutter}
	\IEEEauthorblockA{\textit{School of Aerospace Engineering} \\
		\textit{Georgia Institute of Technology}\\
		Atlanta, GA, USA \\
		lsutter6@gatech.edu}\\
	\IEEEauthorblockN{Fran\c{c}ois Neumann}
	\IEEEauthorblockA{
		\textit{Safran Electronics \& Defense}\\
		Massy, France \\
		francois.neumann@safrangroup.com}
	\and
	\IEEEauthorblockN{Eric Feron}
	\IEEEauthorblockA{\textit{School of Aerospace Engineering} \\
		\textit{Georgia Institute of Technology}\\
		Atlanta, GA, USA \\
		eric.feron@aerospace.gatech.edu}
	}

\maketitle

\IEEEpeerreviewmaketitle 

%% file: subfiles/introduction.tex
The onset of multi-core processors appeared as a golden opportunity for the embedded
systems industry to improve efficiency of embedded computers. Multicore
processors carry several benefits over single core ones, bringing more computational
power through parallelization without increasing chip's internal frequency,
and without increased energy consumption or increased heating. They
now pervade cellular communication devices and embedded electronics for mass-market,
for example, and many other industries are now taking advantage of
such processors, such as 
the 
automotive industry \cite{Multicore_Automotive}, 
the 
biotechnology industry \cite{Multicore_BioTech} 
and the 
circuit industry \cite{Multicore_Circuit}
.
However, as far as critical systems are concerned, these benefits come with
great certification challenges \cite{nowotsch2012leveraging} \cite{reichenbach2010multi}, since parallel applications on a multi-core processor
may interfere. The aerospace industry is yet undertaking to take up this
challenge \cite{durak2019introduction}.\\ 

A reconfigurable multi-core architecture that could host safety critical applications, e.g. \cite{REDEFINE:1}, \cite{5695561}, \cite{recore},
 can become an example of a safe multi-core processor by taking
advantage of the inherent redundancy of such processors that enables graceful
degradation \cite{kinnan2009}: when some core fails, we can use the multiple remaining ones by
reallocating affected applications to a healthy area of the chip.\\

The inherent redundancy in such parallel architecture can also be seen as an opportunity to increase the reliability of computing systems, be it in safety critical embedded systems or for computing centers requiring guaranties of continuity of service. \\

For example, several attempts have been made to increase the reliability of safety-critical systems using multi-core processors. In \cite{ABB_symm_MC}, an ``hypervisor" is used to organize access to shared resources for applications, including safety-critical ones. However, a failure of this hypervisor is not taken into account in this patent. Therefore, such technique just moves the problem since the whole reliability is carried by the reallocation decision organ, which constitutes a single point of failure: the most complex and efficient reallocator is pointless if the system it executes on fails. 
In \cite{ABB_IAC_conf}, backup allocations are pre-calculated for each failure case and they are stored by individual Computational Units (CUs). For small architecture with only a few CUs, this solution is satisfactory and ensures a continuous fault tolerance of the system without requiring a centralized allocator. However, storing backup configuration can require a lot of memory when the architecture becomes bigger. Also, the proposed approach does not consider application that can themselves be parallelized and executed on several CUss at the same time. \\

Our approach differs from these two solutions by providing an on-line and decentralized reallocation algorithm for a general architecture that can be represented by a graph and for parallelized applications requiring several CUs to execute. \\

Even though this work is motivated by a multi-core architecture, it presents a decentralized task allocation algorithm for an abstract parallel computing architecture made of a set of CUs connected together and forming a network. Such an architecture can represent for example a multi-core processor, 
with each CU standing for one core, a cluster of high-performance computers, or a team of mobile robots. The aim of the algorithm is to find the optimal allocation of an a priori defined set of tasks on the architecture while taking into account the faults affecting the CUs. The faults are assumed to be detected by the algorithm when they occur either via a timeout mechanism or a voter, but this work does not provide details of those fault detection mechanisms. As described later, two types of fault will be considered, the first one completely stopping the operation of the CU, and a second one considered to modify the computed output of the CU.\\

The second main feature of this work is the decentralized aspect of the allocation process. Decentralized means here that there is no central element deciding alone of the allocation for the rest of the architecture. Instead, we use redundant copies of the allocation algorithm executed on the architecture itself, meaning that the copies must reallocate themselves. This is achieved by using majority voting systems.\\ 

This work also presents an experimental setup reproducing several aspects of a parallel computing architecture and used to implement the proposed decentralized allocation algorithm. The setup uses a network of Raspberry Pi single board computers \cite{pi} to represent the CUs of the architecture. \\

%% file: subfiles/math_formulation.tex
This section describes the mathematical formulation of the general allocation problem that is considered in this work. The idea is to use this mathematical formulation in an Integer Linear Program (ILP), whose solution is the best allocation of the tasks on the parallel computing platform (multi-core processors, network of computers in a computing center, etc), according to criteria described in Section \ref{subsubsec:obj_func}, taking into account the number of applications running, their priority, the number of reallocated applications and the length of communication paths between allocators and other applications. \\


The considered parallel computing platform is represented by a 
directed simple graph $\mathcal{G} = (V,E)$, where $V$ is the set of \textit{vertices} and $E \subseteq \left\{ \left( x, y\right) \in V^2 \mid x \neq y \right\}$ is the set of \textit{edges} \cite{mesbahi2010graph}. Each vertex of $\mathcal{G}$ represents a CU, 
for example one core in a multi-core processor or at a different scale, one computer in a massively parallel supercomputer, 
and each edge of $\mathcal{G}$ represents a physical communication link between two CUs. The communication links are considered bidirectional, and therefore the orientation of edges can be chosen arbitrarily: we choose them to be oriented only to write more conveniently further constraints on the communication flow. \\

The graph $\mathcal{G}$ therefore represents the topology of the platform. For example, the platform can have a simple square mesh topology, as represented in Fig. \ref{fig:blank_noc}. \\

\begin{figure}[h]
	\captionsetup{justification=centering}
	\centering
	\includegraphics[width=0.45\linewidth]{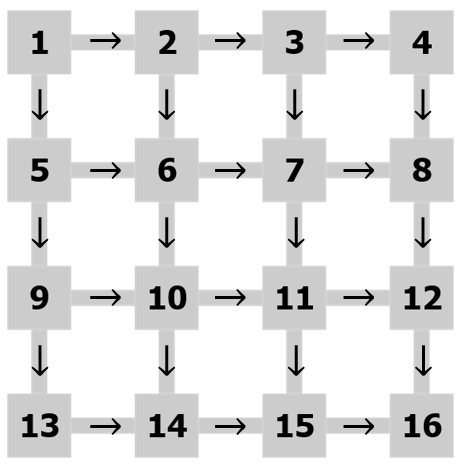}
	\caption{Example of square mesh topology. Orientation of edges are arbitrary.}
	\label{fig:blank_noc}
\end{figure}

From $\mathcal{G}$, we define parameters that will be used later in this work.
\begin{definition}
	$N_{\textrm{CUs}}$ is defined as the number of CUs in the computing platform, that is the number of vertices of $\mathcal{G}$.
\end{definition}	


\begin{definition}
	$N_{\textrm{paths}}$ is defined as the number of Physical Communication Links, or physical paths, in the platform, that is the number of edges of $\mathcal{G}$.
\end{definition}
\vspace{1em}

Let $N_{\textrm{app}} \in \mathbb{N}$ and $\mathcal{A} = \left\{ \mathrm{app}_{k} ,\ k \in \llbracket 1, N_{\textrm{app}} \rrbracket \right\}$ be a set of applications to be executed on the parallel computing platform. The applications in $\mathcal{A}$ are ranked by priority, $\mathrm{app}_{1}$ having the highest priority and $\mathrm{app}_{N_{\textrm{app}}}$ having the lowest one. The ranking is established \textit{a priori} and represents the tolerated order in which we stop applications in case of computing resource failures. 
In the context of a commercial aircraft, an example of such applications with different priority would the engine controller, with the highest priority, and a health monitoring application, with a lower priority, which is in charge of analyzing data from the engine in order to estimate its wear and to predict when maintenance operations are required. In case of computing resource failures, it would be tolerated in this context to stop the health monitoring application in order to maintain the execution of the engine controller. \\



For $k \in \llbracket 1, N_{\textrm{app}} \rrbracket$, we assume that the compiler for the considered architecture decomposes the application $\mathrm{app}_{k}$ into a undirected simple graph $\mathcal{G}_{k} = (V_k,E_k)$, where each vertex, that we will call \textit{Application Node}, represents a sub-task of $\mathrm{app}_{k}$ that must be executed by a CU, and each edge represents a required communication link between two Application Nodes, that we will call an \textit{Application Link}. Fig. \ref{fig:example_appli} gives an example of such application graphs.\\

\begin{figure}[h]
	\captionsetup{justification=centering}
	\centering
	\includegraphics[width=0.8\linewidth]{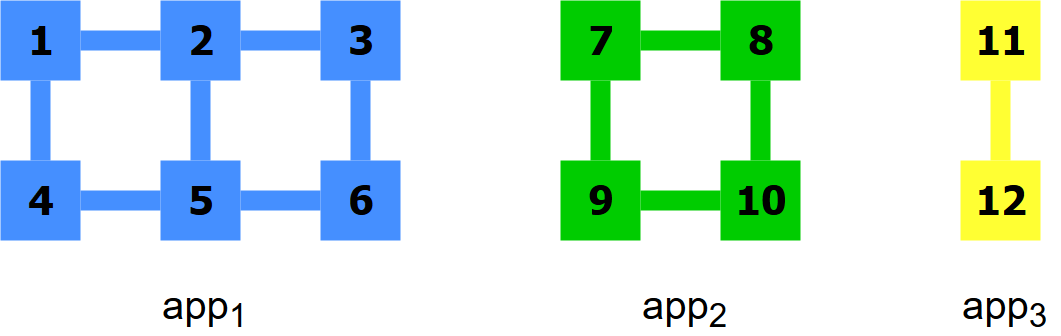}
	\caption{Example of application graphs. Each application node is identified with a unique index.\\
	$\mathrm{app}_{1}$ has highest priority, $\mathrm{app}_{3}$ has the lowest.}
	\label{fig:example_appli}
\end{figure}


From each graph $\mathcal{G}_{k}$ for $k \in \llbracket 1, N_{\textrm{app}} \rrbracket$, we define the following parameters.

\begin{definition}
	$N^k_{\textrm{nodes}}$ is defined as the number of Application Nodes in application $k$ and $N^k_{\textrm{links}}$ is defined as the number of Application Links in application $k$
\end{definition}

\begin{definition}	
	$N_{\textrm{nodes}} \coloneqq \sum_{k=1}^{N_{\textrm{apps}}} N^{k}_{\textrm{nodes}}$ is the total number of Application Nodes, and $N_{\textrm{links}} 
	\coloneqq \sum_{k=1}^{N_{\textrm{apps}}} N^{k}_{\textrm{links}}$ is the total number of Application Links.\\
\end{definition} 
\vspace{0.5em}

Each application node is given a global index $j \in \llbracket 1, N_{\textrm{nodes}} \rrbracket$ with the following procedure: the nodes of $\mathrm{app}_{1}$ keep the same indices as in the local numbering of vertices in $\mathcal{G}_{1}$ ; then the global indices for nodes of $\mathrm{app}_{2}$ are obtained by increasing their local indices by $N^1_{\textrm{nodes}}$ ; and so on for the nodes of $\mathrm{app}_{k}$, by increasing the local numbering by $\sum_{l=1}^{k-1} N^{l}_{\textrm{nodes}}$. The result of the global numbering of the nodes can be seen on Fig. \ref{fig:example_appli}. An identical process is applied to obtain a global numbering of the edges of the application graphs. \\

The problem that we tackle here is to assign applications to CUs of the architecture while faults affect some CUs, taking into account the priority of the applications and specific constraints of the architecture. A solution will look like Fig. \ref{fig:example_sol}. The approach that we take here to solve the problem is to formulate the allocation problem as Integer Linear Program (ILP) and use a state-of-the-art IP solver such as ``GNU Linear Programming Kit" (GLPK) \cite{glpk}.\\

\begin{figure}[h]
	\captionsetup{justification=centering}
	\centering
	\includegraphics[width=0.45\linewidth]{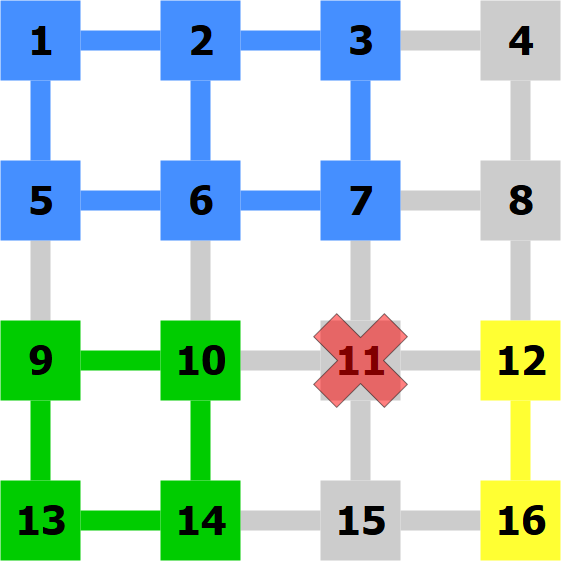}
	\caption{Example of a solution with a fault on CU 11.}
	\label{fig:example_sol}
\end{figure}


An additional aspect of the problem that we propose to solve is to make the allocation process decentralized, in the sense detailed in the introduction and in Section \ref{sec:decentr}, with no central computing element allocating the tasks according to the solution of the ILP problem. The way this decentralized allocation is achieved is specifically described in Section \ref{sec:decentr}: it involves several copies of the task computing the allocation and being executed on the platform itself. The number of such copies is the last parameter of our problem.

\begin{definition}
	$N_{\textrm{realloc}}$ is defined as the number of copies of the $\textit{Allocator Application}$.
\end{definition}

\vspace{1em}

The next section details how the allocation problem is formulated as an ILP problem. \\






%% file: subfiles/allocation_algo.tex
\subsection{Matrix representation of graphs}

\begin{definition}
	From the graph representation $\mathcal{G} = (V,E)$ of the parallel computing platform, the $N_{\textrm{CUs}} \times N_{\textrm{paths}}$ \textit{incidence matrix} $G$ associated with $\mathcal{G}$ is defined as:
	\[[G]_{ij} \coloneqq \begin{cases} 
	-1 & \text{if $e_j \in E$ leaves $v_i \in V$} \\
	1 & \text{if $e_j \in E$ enters $v_i \in V$} \\
	0 & \text{otherwise} 
	\end{cases}\]
	And the $N_{\textrm{CUs}} \times N_{\textrm{paths}}$ \textit{NoC unoriented incidence matrix} $\hat{G}$ associated with $\mathcal{G}$ is defined as:
	\[[\hat{G}]_{ij} \coloneqq \left|[G]_{ij}\right|\]
\end{definition}
\vspace{1em}

\begin{definition}
	From the graph $\mathcal{G}_k = (V_k,E_k)$ representing the $k$-th application, the $N^k_{\textrm{nodes}} \times N^k_{\textrm{links}}$ \textit{application unoriented incidence matrix} $H^k$ associated with $\mathcal{G}_k$ is defined as:
	\[[H]^k_{ij} \coloneqq \begin{cases} 
	1 & \text{if $v^k_i \in V_k$ and $e^k_j \in E_k$ are incident} \\
	0 & \text{otherwise} 
	\end{cases}\]
	Furthermore, the $N_{\textrm{nodes}} \times N_{\textrm{links}}$ \textit{overall application unoriented incidence diagonal block-matrix} $H$ is defined as
	\[
	H \coloneqq \begin{bmatrix} 
	H^{1} &   &   \\
	& \ddots &   \\
	&   & H^{k} 
	\end{bmatrix}
	\]
\end{definition}
\vspace{1em}

\subsection{Definition of the Decision Variables}\label{subsubsec:decision_vars}

\begin{definition}
	The $N_{\textrm{CUs}} \times N_{\textrm{nodes}}$ \textit{decision matrix} $X^{\textrm{CUs} \to \textrm{nodes}}$, mapping Application Nodes to CUs, is defined as:
	\[X^{\textrm{CUs} \to \textrm{nodes}}_{ij}=\begin{cases} 
	1 & \parbox[t]{.51\linewidth}{if the CU $i$ is allocated to the Application Node $j$}\\
	0 & \text{otherwise} 
	\end{cases}
	\]
\end{definition}

\begin{definition}
	The $N_{\textrm{paths}} \times N_{\textrm{links}}$ \textit{decision matrix} $X^{\textrm{paths} \to \textrm{links}}$, mapping Application Links to 
	Physical Links, is defined as:
	\[X^{\textrm{paths} \to \textrm{links}}_{ij}=\begin{cases} 
	1 & \parbox[t]{.51\linewidth}{if the Physical Link $i$ is allocated to the Application Link $j$}\\
	0 & \text{otherwise} 
	\end{cases}
	\]
\end{definition}

\begin{definition}
	The $N_{\textrm{apps}} \times 1$ \textit{decision vector} $r$, representing which applications are executed, is defined as:
	\[r_i=\begin{cases} 
	1 & \parbox[t]{.51\linewidth}{if the application $i$ is running}\\
	0 & \text{if it is dropped} 
	\end{cases}
	\]
\end{definition}

\begin{definition}
	The $N_{\textrm{nodes}} \times 1$ \textit{decision vector} $M$, representing which application nodes are reallocated, is defined as:
	\[M_i=\begin{cases} 
	1 & \parbox[t]{.6\linewidth}{if the Application Node $i$ is moved from its previously allocated CU}\\
	0 & \text{otherwise} 
	\end{cases}
	\]
\end{definition}

\begin{definition}
	For $k \in \llbracket 1, N_{\textrm{app}} \rrbracket$, the $N_{\textrm{paths}} \times N_{\textrm{CUs}}$ \textit{decision matrix} $X^{\textrm{Comm, } k}$, representing communication paths between the $k$-th allocator application and every CU of the platform, is defined as:
	\[X^{\textrm{Comm, } k}_{ij}=\begin{cases} 
	-1 & \parbox[t]{.51\linewidth}{if the Physical Link $i$ is used to communicate between the allocator $k$ and the CU $j$ in the negative direction}\\
	1 & \parbox[t]{.51\linewidth}{if the Physical Link $i$ is used to communicate between the allocator $k$ and the CU $j$ in the positive direction}\\
	0 & \text{otherwise} 
	\end{cases}
	\]
	Positive (respectively negative) direction means that the communication takes place in the same (respectively opposite) direction as the edge of the directed graph $\mathcal{G}$, as in Fig. \ref{fig:blank_noc}.
\end{definition}

\vspace{1em}

\subsection{Formulation of the optimization model}\label{subsec:opt_prob}

This section gives the detail of the formulation of the optimization problem that is solved each time a new fault is detected. This formulation includes the detail of the chosen objective function and constraints.\\


\subsubsection{General form of the optimization model}

The allocation problem is formulated as an Integer Linear Program (ILP) of the form \cite{HillLieb01}:

\begin{equation}\label{eq:general_formulation}
	\begin{tabular}{l l }
		maximize & $f(\boldsymbol{x}) = \boldsymbol{c}^\mathrm{T} \boldsymbol{x}$ \\ 
		subject to & $M_1 \boldsymbol{x} \leq \boldsymbol{b_1}$\\ 
		& $M_2 \boldsymbol{x} = \boldsymbol{b_2}$\\
		and & $\boldsymbol{x}$ is a vector of integers.   
	\end{tabular}
\end{equation}

$\boldsymbol{x}$ is the global vector of decision variables derived from the vectorization and the aggregation of the decision matrices from Section \ref{subsubsec:decision_vars}. We define it formally as:

\begin{align}\label{eq:decision_vector}
\boldsymbol{x} &= 
\begin{pmatrix}
\mathrm{vec}( X^{\textrm{CUs} \to \textrm{nodes}})\\
\mathrm{vec}( X^{\textrm{paths} \to \textrm{links}})\\
r\\
M\\
\mathrm{vec}( X^{\textrm{Comm, } 1})\\
\vdots\\
\mathrm{vec}( X^{\textrm{Comm, } N_{\textrm{realloc}}})\\
\end{pmatrix}
\end{align}\\
where $\mathrm{vec}$ is the common vectorization function for matrices:
$\forall \ Q = (q_{i, j})_{1 \leq i \leq m ,\ 1 \leq j \leq n},\\
\mathrm{vec}(Q) = [q_{1,1}, \ldots, q_{m,1}, q_{1,2}, \ldots, q_{m,2}, \ldots, q_{1,n}, \ldots, q_{m,n}]^\mathrm{T}$.\\

$\boldsymbol{c}$ is the coefficients of the objective function and $M_1$, $M_2$, $\boldsymbol{b_1}$ and $\boldsymbol{b_2}$ are parameters derived from the aggregation of the constraints of the problem that are described in the following sections. For example, for each scalar inequality constraint, after arranging the inequality with all decision variables on the left-hand side in the same order as in $\boldsymbol{x}$ and constant terms on the right-hand side, a row containing the coefficients of the decision variables is added to $M_1$ and the constant term is added in the vector $\boldsymbol{b_1}$. The same is done for equality constraints to build $M_2$ and $\boldsymbol{b_2}$.\\


\subsubsection{Objective function}\label{subsubsec:obj_func}
Given the priority of the applications in an ascending order i.e. the first application has the highest priority and the $N_{\textrm{apps}}$-th application has the lowest one, the objective function is used in order to maximize the number of executed applications while minimizing the number of reallocations and the length of communication paths. The chosen objective function, in terms of $\boldsymbol{x}$ as defined above in equation \ref{eq:decision_vector}, is:


\begin{equation}\label{eq:obj_func}
\begin{split}
\max 
\Bigg\{
f(\boldsymbol{x}) = \sum_{k = 1}^{N_{\textrm{apps}} } \alpha_k \cdot r_k &- (\beta + 1)\sum_{j = 1}^{N_{\textrm{nodes}}} M_j \\
&- \sum_{k = 1}^{N_{\textrm{realloc}}}\sum_{j = 1}^{N_{\textrm{CUs}}}\sum_{i = 1}^{N_{\textrm{paths}}} \left|X^{\textrm{Comm, } k}_{ij}\right| 
\Bigg\} ,
\end{split}
\end{equation}
where 
\begin{equation}\label{eq:obj_func_coeff}
\begin{split}
\beta &= N_{\textrm{realloc}} \times N_{\textrm{CUs}} \times N_{\textrm{paths}} \\
\alpha_{N_{\textrm{apps}}} &= 
(\beta+1) \times N_{\textrm{nodes}} + \beta + 1 \\
\text{and } \forall k < N_{\textrm{apps}}: \\
\alpha_k &= \sum_{l = k+1}^{N_{\textrm{apps}}} \alpha_l + (\beta+1) \times N_{\textrm{nodes}} + \beta + 1 .
\end{split}
\end{equation}\\

The coefficients of the objective function are chosen to prioritize the different aspects that are optimized in this function.
\begin{enumerate}
	\item The first priority is to execute each application, even if it means more reallocations and longer communication paths.
	\item Then, minimizing the number of reallocations is more important than having shorter communication paths, since a reallocation temporarily interrupts the execution of the allocation. 
	\item When running all applications is not feasible, the priorities of the applications are enforced and executing any given application is more important than running any number of applications with a lower priority. However, if because of its geometry, a given application cannot be executed anyway, nothing prevents lower-priority applications from being executed.\\
\end{enumerate}

These requirements motivated the choice for the coefficients in the objective function. The proof that these coefficients allow the objective function to meet these requirements is given in Appendix. \\ 



Note that the problem of minimizing or maximizing the absolute value of the $X^{\textrm{Comm, } k}_{ij}$ variables, which is a nonlinear program, can be reformulated as a linear program by introducing additional variables and constraints \cite{HillLieb01}, that were not presented in the previous section for conciseness. For each entry $X^{\textrm{Comm, } k}_{ij}$ of $X^{\textrm{Comm}}$, an auxiliary variable $\hat{X}^{\textrm{Comm, } k}_{ij}$ is introduced to represent its absolute value, and two extra constraints are added: 
$$ + X^{\textrm{Comm, } k}_{ij} \leq \hat{X}^{\textrm{Comm, } k}_{ij}, $$
$$- X^{\textrm{Comm, } k}_{ij} \leq \hat{X}^{\textrm{Comm, } k}_{ij}. $$
$\hat{X}^{\textrm{Comm, } k}_{ij}$ is then used instead of $\left|X^{\textrm{Comm, } k}_{ij}\right|$ in the objective function. Because the objective function tends to maximize $- \left|X^{\textrm{Comm, } k}_{ij}\right|$, so to minimize $\hat{X}^{\textrm{Comm, } k}_{ij}$, one of the two previous constraints will be binding, the stricter one, where the left-hand side is the greatest and equal to $\max(+ X^{\textrm{Comm, } k}_{ij}, - X^{\textrm{Comm, } k}_{ij})$, which is exactly $\left|X^{\textrm{Comm, } k}_{ij}\right|$. The other constraint will be non-binding and therefore does not affect the optimal point. It thus ensures that $\hat{X}^{\textrm{Comm, } k}_{ij}$ is equal to $\left|X^{\textrm{Comm, } k}_{ij}\right|$. \\


\subsubsection{Constraints}

\paragraph{Domain of decision variables}
The decision variables $X^{\textrm{CUs} \to \textrm{nodes}}$, $X^{\textrm{paths} \to \textrm{links}}$, $r$ and $M$ are binary i.e. the value of their entries must be either 0 or 1.\\

The entries of $X^{\textrm{Comm, } k}$ for $k \in \llbracket 1,\ N_{realloc} \rrbracket$ must belong to $\{-1,\ 0,\ 1\}$.\\

\paragraph{Resource allocation and partitioning}
Several equations express the constraints of allocating the resources of the CUs to applications while enforcing partitioning on the platform.\\

\begin{itemize}

\item Each CU can be allocated to at most one application, as a way to enforce spatial partitioning of applications on the platform, i.e.
\begin{equation}
\forall i \in \llbracket 1, N_{\textrm{CUs}} \rrbracket,\ \sum_{j=1}^{N_{\textrm{nodes}}}  X^{\textrm{CUs} \to \textrm{nodes}}_{ij} \leq 1 .
\end{equation}

\item Each running Application Node must be assigned to exactly one CU, i.e. 
\begin{equation}
\forall i \in \llbracket 1, N_{\textrm{nodes}} \rrbracket,\ \sum_{j=1}^{N_{\textrm{CUs}}}  X^{\textrm{CUs} \to \textrm{nodes}}_{ji} = r_{N(i)} .
\end{equation}
$N(i)$ is the application number corresponding to Application Node $i$.\\

\item A physical communication link of the platform can be allocated to at most one Application Link\footnote{This does not mean that this communication link cannot be used for other communication purposes on the architecture, but only one of the Application Link computed by the compiler for the applications can be allocated to that physical communication link.}, i.e.
\begin{equation}
\forall i \in \llbracket 1, N_{\textrm{paths}} \rrbracket,\ \sum_{j=1}^{N_{\textrm{links}}}  X^{\textrm{paths} \to \textrm{links}}_{ij} \leq 1 .
\end{equation}

\item Each running Application Link must be assigned to exactly one physical communication link of the platform, i.e. 
\begin{equation}
\forall i \in \llbracket 1, N_{\textrm{links}} \rrbracket,\ \sum_{j=1}^{N_{\textrm{paths}}}  X^{\textrm{paths} \to \textrm{links}}_{ji} = r_{L(i)} .
\end{equation}
$L(i)$ is the application number corresponding to Application Link $i$.\\

\end{itemize}

\paragraph{Compliance with the platform}
An Application link, adjacent to an Application Node that has been mapped to a given CU, must be allocated to a Physical Link that is adjacent to that CU, i.e. \\
\begin{equation} \label{eq:compli1}
X^{\textrm{CUs} \to \textrm{nodes}}\ H = \hat{G}\ X^{\textrm{paths} \to \textrm{links}} .
\end{equation}

This equation (\ref{eq:compli1}) is equivalent to the scalar equations (\ref{eq:compli2}): 
\begin{equation}\label{eq:compli2}
\forall i \in \llbracket 1, N_{\textrm{CUs}} \rrbracket,\ \forall j \in \llbracket 1, N_{\textrm{links}} \rrbracket,$$ $$\sum_{k=1}^{N_{\textrm{nodes}}} X^{\textrm{CUs} \to \textrm{nodes}}_{ik} \ H_{kj} = \sum_{l=1}^{N_{\textrm{paths}}} \hat{G}_{il} \ X^{\textrm{paths} \to \textrm{links}}_{lj}.
\end{equation}
The left-hand side is equal to one if and only if the CU $i$ has been allocated to Application Node $k$ and Application Node $k$ is adjacent to Application Link $j$. The right-hand side is equal to one if and only if the CU $i$ is adjacent to the Physical Link $l$ and the Physical Link $l$ is allocated to Application Link $j$, which proves the correctness of the constraint.\\





\paragraph{Reallocating several applications}
A given Application Node can either remain affected to the same CU, either be moved or be dropped:
\begin{equation} \label{eq:reallc_several_apps}
\begin{split}
\forall i \in \llbracket 1, N_{\textrm{CUs}} \rrbracket ,\ \forall j \in \llbracket 1, N_{\textrm{nodes}} \rrbracket ,\ 
\text{s.t. } X^{\textrm{CUs} \to \textrm{nodes}}_{\textrm{old}\ ij} = 1,\\ \left(1- r_{N(j)}\right) + M_{j} + X_{ij}^{\textrm{CUs} \to \textrm{nodes}} = X^{\textrm{CUs} \to \textrm{nodes}}_{\textrm{old}\ ij},
\end{split}
\end{equation}
with $X^{\textrm{CUs} \to \textrm{apps}}_{\textrm{old}}$ be the parameter containing the mapping between CUs and Application Nodes computed during the previous allocation.\\

This constraint (\ref{eq:reallc_several_apps}) is ignored for the initial allocation.\\

\paragraph{Faults}
We assume the parallel computing platform is equipped a fault detection system that can detect and inform the allocators when a CU fails. From this information, we can add constraints to take into accounts fault in the platform. Within a CU $i$: \\
\begin{itemize}
	\item If the CU is healthy, any Application Node can be mapped on the CU.
	\item If the CU is faulty, then no Application Nodes can be mapped on the CU $i$:
	\begin{equation}
		\sum_{k=1}^{N_{\textrm{nodes}}}  X^{\textrm{CU} \to \textrm{apps}}_{ik} = 0 .
	\end{equation}
\end{itemize}

The detection of this fault is either assumed for the model or detected by the voter using the majority rule described in \ref{subsec:soft_comp}.\\

\paragraph{Communication constraints} \label{par:comm_constr}
Since the allocators are executed on the platform, we must ensure that they will be able to send the allocation they computed to the other CUs of the platform, given the communication links that allows each CU to send a message only through its neighbors. Therefore, we must make sure that there exists a path from each allocator to the other CUs.\\
\begin{equation}
\forall k \in \llbracket 1, N_{\textrm{realloc}} \rrbracket ,\ G\ X^{\textrm{Comm, } k} = S^k
\end{equation}
where $S^k$ is the $N_{\textrm{CUs}} \times N_{\textrm{CUs}}$ \textit{source-sink matrix} $S^k$, depending on $X^{\textrm{CUs} \to \textrm{nodes}}$ and defined by: 

\[[S]^k_{ij} \coloneqq \begin{cases} 0 & \text{if $\mathrm{deg}(v_i) = 0$ in $\mathcal{G}$} \\
0 & \text{if CU $i$ is faulty} \\
-X^{\textrm{CUs} \to \textrm{nodes}}_{i \ \textrm{node\_of\_alloc}(k)} + [I_{N_{\textrm{CUs}}}]_{ij} & \text{otherwise} \end{cases}.\]\\
where the degree of a vertex $\mathrm{deg}(v_i)$ is the number of edges connected to it, and $\textrm{node\_of\_alloc}(k)$ is the Application Node corresponding to allocator $k$. When the CU is neither faulty nor without any neighbor, in each path between an allocator and a given CU $i$, the allocator is the source (-1) and the CU $i$ is the sink (+1). \\

\paragraph{Constraints specific to the architecture}
Additional constraints can be added to respect specific aspects of the considered architecture. \\

For example, some multi-core architectures \cite{REDEFINE:1},
where intra-application communication between CU can happen only in a specific way as 
illustrated in Fig. \ref{orientation}, orientation of the applications on the architecture matters because nodes that can communicate in a given orientation will not be able to do so if they are rotated on the architecture. Therefore the orientation as computed by the compiler must be enforced. \\

To ensure correct orientation of applications, another set of constraints 
is also needed. In order to enforce this, the numbering of the CUs on the platform is used. For example, as illustrated in Fig. \ref{orientation}, a CU has always a number difference of $-1$ with its right neighbor and $+N_{\textrm{row}}$ with its top neighbor, where $N_{\textrm{row}}$ is the number of Tiles per row of the NoC ($N_{\textrm{row}}=4$ in our example).\\

The difference between the numbers of the contiguous pairs of CUs allocated to an application must match the orientation computed by the compiler. \\

Let $j^k$ be the index of the top-left node of the $k$-th application:

\begin{equation}
\begin{split}
\forall i \in \llbracket 1, N_{\textrm{CUs}} \rrbracket ,\ 
\forall k \in \llbracket 1, N_{\textrm{apps}} \rrbracket ,\\
X^{\textrm{CUs} \to \textrm{nodes}}_{ij^k} = X^{\textrm{CUs} \to \textrm{nodes}}_{(i+1)(j^k+1)}\\
X^{\textrm{CUs} \to \textrm{nodes}}_{ij^k} = X^{\textrm{CUs} \to \textrm{nodes}}_{(i+N_{\textrm{row}})(j^k+N_{\textrm{row}}^k)}.
\end{split}
\end{equation}


where $N_{\textrm{row}}^k$ is the number of nodes per row of the $k$-th application.\\

\begin{figure}[h]
	\captionsetup{justification=centering}
	\centering
	\includegraphics[width=0.75\linewidth]{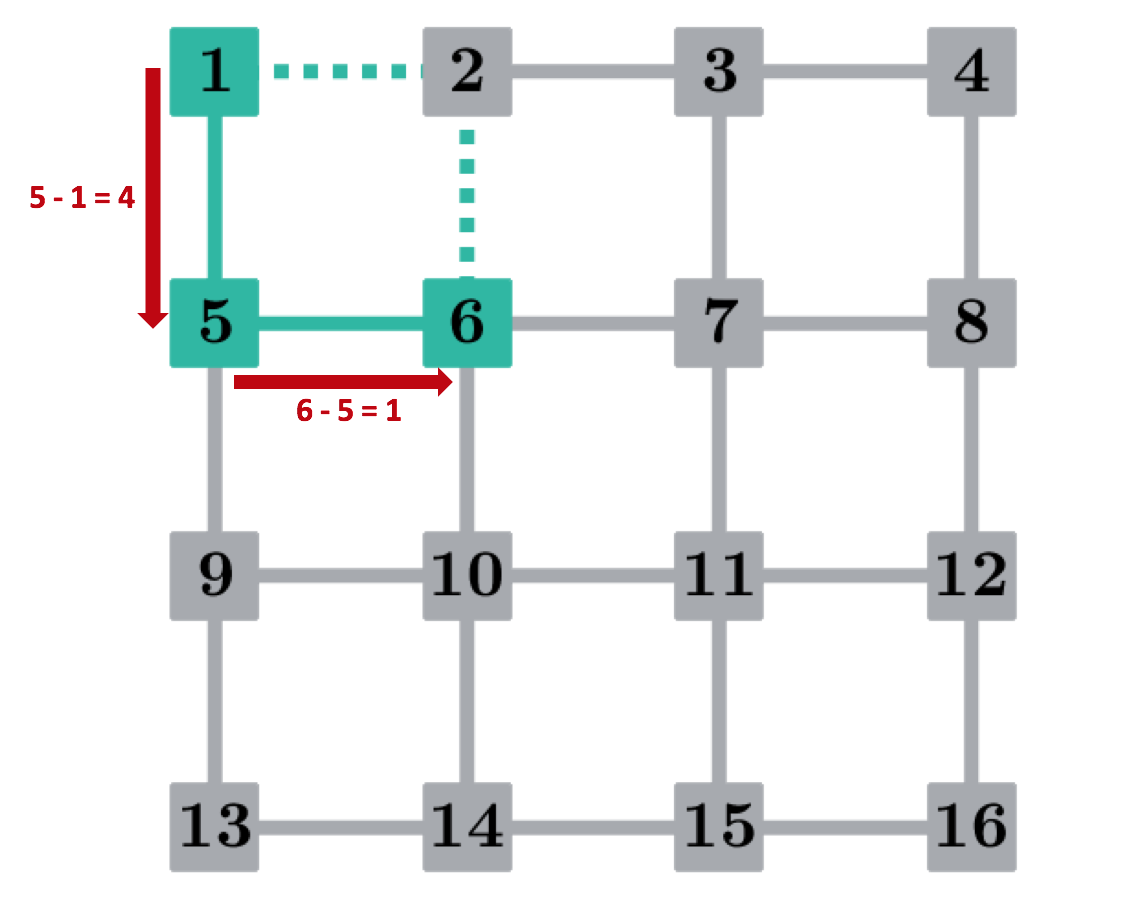}
	\caption{By equating the difference between two CUs' indices allocated to an application to a specific number, the spatial orientation of the application can be enforced.}
	\label{orientation}
\end{figure}

%% file: subfiles/decentralization.tex

In this paper, we use the word \textit{decentralized} to qualify a system where no single CU has control over all the other ones in the parallel architecture: there is no central CU whose failure jeopardizes the operation of the whole parallel architecture. In safety words, this means that \textit{no CU constitutes a single point of failure}.\\

We focus here on CUs, but there are other elements that may be a single point of failure and that we do not take into account in this work. For example, electrical power may be provided by one unique and central power supply unit, which is an obvious single point of failure if not designed carefully. To mitigate the effect of other single point failures, methods for safety assessment process may be conducted \cite{sae1996guidelines}.\\


\subsection{N-modular redundancy and majority voting system}\label{sub:nmr_vote}

To develop a decentralized allocation system for the considered parallel computing platform, we chose to use the concept of N-modular redundancy with a majority voting system \cite{TMR}.\\

In this approach, $N_{\textrm{realloc}}$ is an odd number greater or equal to $3$, and $N_{\textrm{realloc}}$ copies of the same sub-tasks are executed in parallel. $N_{\textrm{realloc}}$ is taken odd to avoid the case where equal number of copies agree on two different results. The copies are fed with the same inputs and their outputs are then sent to a majority voting system. As illustrated in Fig. \ref{fig:vote}, the voting system compares the outputs of the redundant copies and filters them: only the result that has been computed by the majority of the redundant copies will be transmitted, i.e. the result computed by at least $ \frac{N_{\textrm{realloc}} + 1}{2}$ redundant copies. The voting system is also used to report the failure of the redundant copies that do not match the majority result. \\

\begin{figure}[!h]
	\captionsetup{justification=centering}
	\centering
	\includegraphics[width=0.6\linewidth]{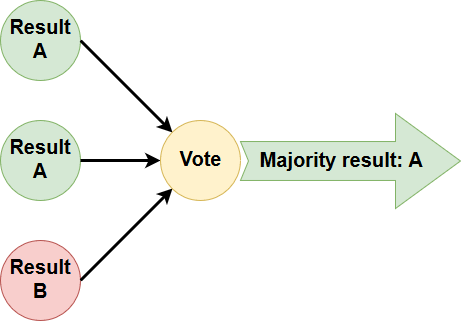}
	\caption{Illustration of the voting process with 3 redundant copies.}
	\label{fig:vote}
\end{figure}


\subsection{Decentralized implementation}


The proposed idea to decentralize the allocation system is to 
execute $N_{\textrm{realloc}}$ modular redundant copies of the allocator application on the architecture itself, with a voting system implemented on each CU. For further examples, $N_{\textrm{realloc}}$ will be taken equal to $3$. \\

In normal conditions, the $N_{\textrm{realloc}}$ copies of the allocator compute the same allocation, since they solve an identical ILP problem, with same inputs and constraints, and because GLPK is a deterministic solver. This allocation is then broadcast to every CU, including the ones executing the allocators. \\

If a CU not running an allocator fails, all 3 allocators compute the same new allocation, in which the affected application is assigned to a new CU, according the algorithm described previously in Section \ref{sec:algo}. This new allocation is then broadcast and received by all CUs. Since the 3 signals that the CUs receive are coherent, they all comply with it and therefore, the affected application is reallocated. \\

On the other hand, as illustrated in Fig. \ref{fig:decentr}, if a CU that was running a copy of the allocator is affected by a fault, the 2 other ones will compute the same new allocation where the affected copy is assigned to a new healthy CU. Regardless of what the faulty allocator computes, only the two coherent allocation sent by the two healthy allocators will be taken into account by the CUs, and the faulty allocator will be reallocated. \\

\begin{figure}[h] 
	\centering
	\captionsetup{justification=centering}
	\begin{subfigure}[t]{0.5\textwidth}
		\captionsetup{justification=centering}
		\centering
		\includegraphics[height=1.8in]{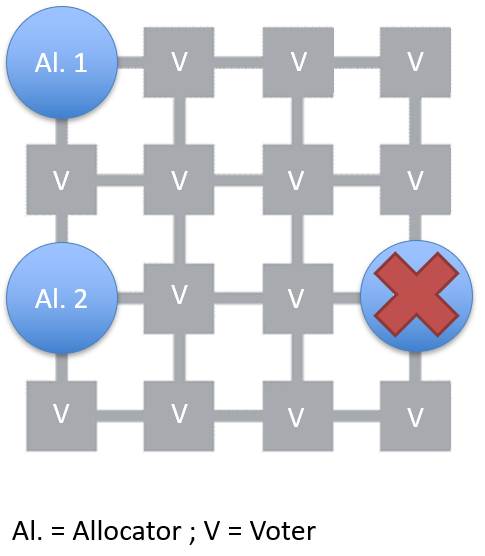}
		\caption{Layout of the allocators on the computing architecture.}
	\end{subfigure}

	\vspace{1em}
	
	\begin{subfigure}[t]{0.5\textwidth}
		\captionsetup{justification=centering}
		\centering
		\includegraphics[height=1.5in]{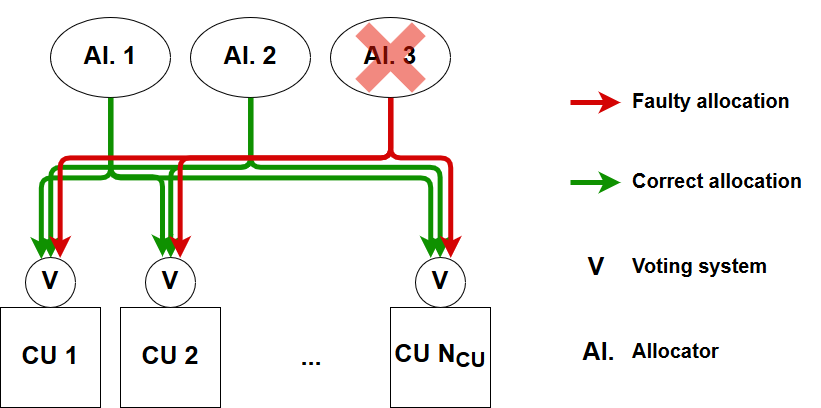}
		\caption{Information flow between allocators and CUs. \\
		The correct allocation includes the instruction for some node $i$ to run allocator 3.}
	\end{subfigure}
	\caption{Fault affecting a CU running an allocator.}
	\label{fig:decentr}
\end{figure}

%% file: subfiles/practical_expl.tex
This section  describes the experimental setup that is used as a representation of the parallel computing platform as well as the result of reallocating safety-critical applications using the previously-described optimization problem.

\subsection{Representation of a parallel computing platform}


\subsubsection{Hardware components}

To illustrate and demonstrate the capabilities of the new formulation of the allocation algorithm in operational conditions, we choose to implement it on a cluster of single-board computer, Raspberry Pi \cite{pi},
in order to control and maintain operation of a physical system despite the presence of faults. \\

%
%
%
%

\paragraph{platform description}


In this setup, a cluster of parallel CUs of $4 \times 4$ units is replicated with a network of 16 Raspberry Pi computers. All of them are connected to a common routing switch in a local area network (LAN). 
Although the use of this common routing switch is a single point failure, it serves a purpose of visualizing that these Raspberry Pi computers are grouped as a single parallel CU. Also, for simplicity of visualization, the network is considered to be a square mesh, instead of a toroidal mesh. \\

One alternative of using a wired LAN network is to use a routing protocol for multi-hop mobile ad hoc network such as \cite{neumann2008better}. This ad-hoc network is implemented on a data link layer, which allows the data transportation protocol operate in a wireless and decentralized fashion as if there is a common routing switch.\\

The goal of this parallel computing platform is to show the possibility to decentralize the allocation process; therefore, there is no central computing unit outside the network and three copies of the allocator are executed on the network, as described in Section \ref{sec:algo}. \\


\paragraph{Faults} \label{subsub:faults}
Two types of faults are considered in this experiment. The first type is computational fault, which randomly affect the computations performed by the Raspberry Pi. We detect this kind of fault by using redundant copies of the considered application combined with a voting system that is described below in section \ref{subsec:soft_comp}.
The second type of fault is assumed to stop the operation of the computing unit it affects. We also assume that this fault can be detected by the network. In practice, each time one of these faults affects a Raspberry Pi, the status signal sent by this Raspberry Pi to the allocators is changed to a signal identifying it as faulty. 
Each of these two kinds of fault can be manually triggered or recovered thanks to a breadboard as seen in Fig. \ref{fig:switch} connected to each Raspberry Pi.\\

\begin{figure}[!h]
	\captionsetup{justification=centering}
	\centering
	\includegraphics[width=0.5\linewidth]{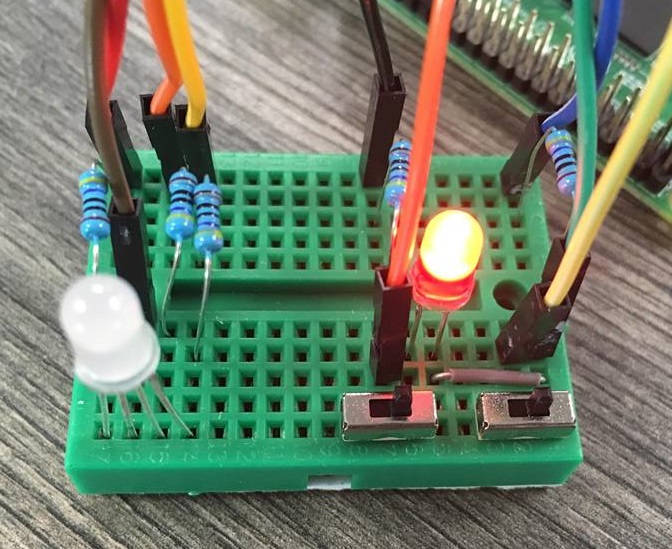}
	\caption{Hardware associated with each Raspberry Pi Tile.\\ The RGB-LED (bottom-left corner) represents the LED application. The red LED (right side) indicates an healthy Tile when turned on. Each switch is used to trigger one type of fault.}
	\label{fig:switch}
\end{figure}

\paragraph{Controlled system}
The physical system we chose to control with this parallel computing platform is a propulsion system, made of an electric fan mounted on a thrust stand. The fan is commanded by using Pulse width modulation (PWM). The measure of the thrust is used by a simple proportional controller executed as a safety-critical application on the platform in order to compute the value of the PWM command required to maintain the thrust at a constant value. \\

\begin{figure}[!h]
	\captionsetup{justification=centering}
	\centering
	\includegraphics[width=1\linewidth]{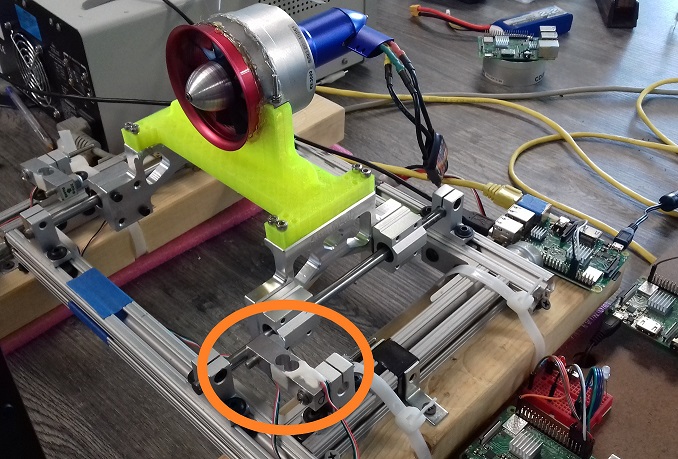}
	\caption{Electric fan mounted on the thrust stand. \\The delivered thrust is measured thanks to a load cell on the stand, indicated by the orange circle.}
	\label{fig:stand}
\end{figure}

An extra Raspberry Pi is used  as the micro-controller of the fan: it converts the value measured by the load cell, sends it to the \textit{Controllers} where the appropriate control value is computed, and generates the corresponding PWM signal controlling the fan. It must therefore be noted that although the same hardware representation is used, this Raspberry Pi does not correspond to the same components as the ones used for the CUs of the platform. \\ 

\subsubsection{Software components}\label{subsec:soft_comp}

Even if a controller is reallocated to healthy Tiles when it is affected by a fault, because of the time required to compute the new allocation and to actually reallocate the set of tasks, the operation of the fan may be temporarily altered during the reallocation process.\\

To avoid interruptions in the operation of the fan during reallocations, we also use a standard Triple Modular Redundancy (TMR) architecture \cite{TMR}. Three copies of the controller are executed on the parallel computing platform. Each one separately computes the duty-cycle value of the PWM signal that should be sent to the fan, given the thrust value that they all receive from the sensor. The three values are sent to the Raspberry Pi representing the micro-controller of the fan, where a voting system decides which control output should be used. 
The vote outputs the result that has been computed by the majority of the controllers, in this case 2 out of 3. Signals are here considered equal if their difference is smaller than a given tolerance. In the case of a fault affecting the output of one of the controllers, the two remaining healthy controllers ensure that the correct value is sent to the fan. The voting system also identifies which controller is not coherent with the two others and informs the allocators of the fault. The reallocation process that we implemented can then take place while providing continuity of service with the two healthy controllers. To complicate the reallocation tasks, each copy of the controller has been arbitrarily attributed to 2 Application Nodes. Concretely, only one of them is responsible of actual computations. \\

Three copies of the allocator execute the allocation algorithm itself. They have second rank priority immediately below the controllers, which represent the safety-critical application in this case. Giving the allocators only the second rank in the priority list can be justified when considering the case where only a controller or an allocator can be executed on the platform: the resource must be allocated to the safety-critical application, in this case the controller, that maintains the operation of the system, whereas the allocator is only a protection against further faults, but cannot alone ensure operation of the controlled system. \\

In addition to these six applications, one dummy application is considered in this experiment: it occupies 2 Tiles of the Fabric, but does not perform actual computation except changing the voltage in the RGB LED to display its corresponding color. It has the lowest priority. \\

Figure \ref{fig:apps_exp} sums up the list of considered applications for the experiment, their relative priority and the resources they require in terms of number of CUs. The initial allocation of these applications on the model is given in Fig. \ref{fig:init_alloc}. \\

\begin{figure}[!h]
	\captionsetup{justification=centering}
	\centering
	\includegraphics[width=0.5\textwidth]{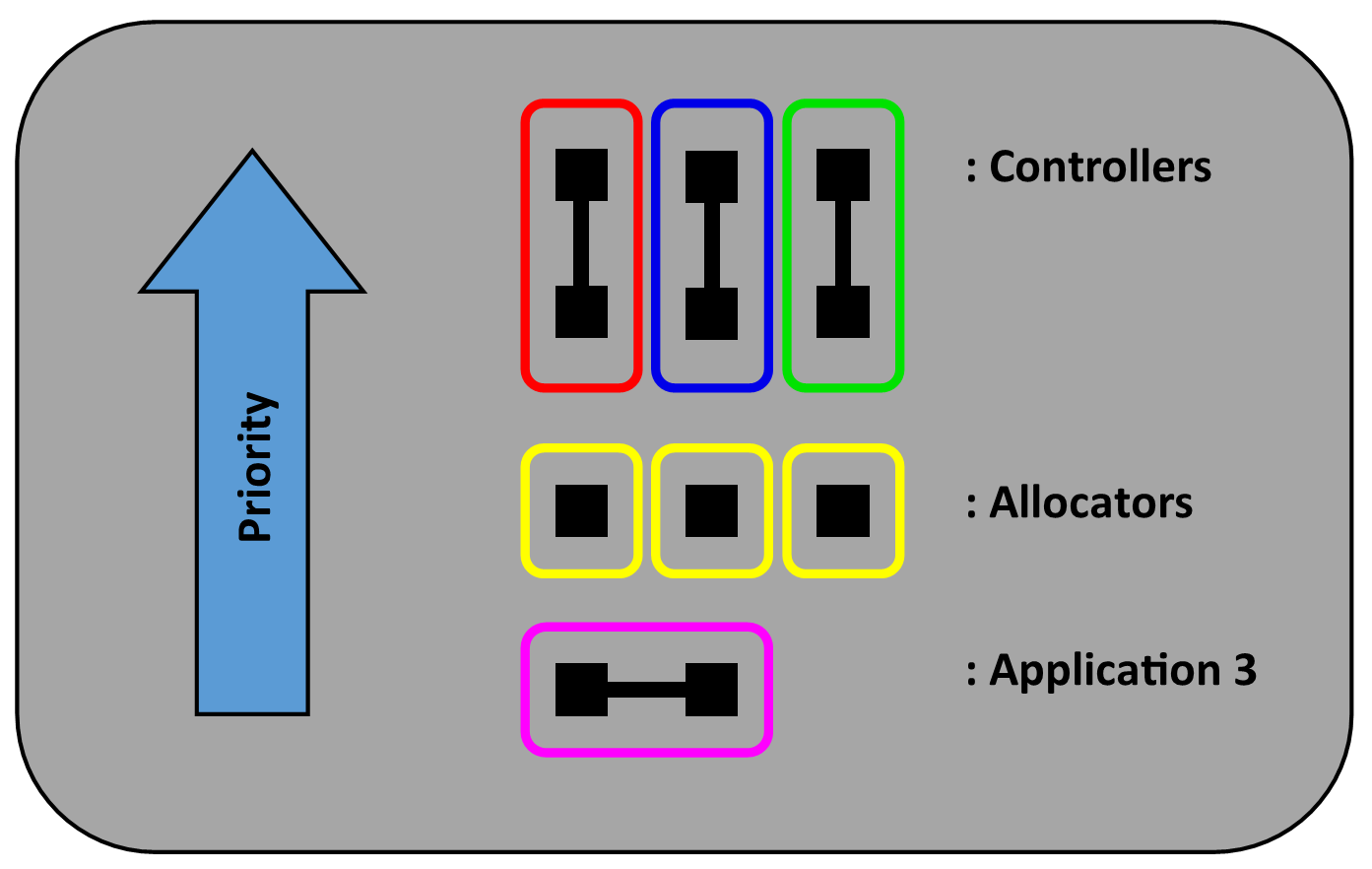}
	\caption{Considered applications for the experiment, their priority and the number of Tiles they require.}
	\label{fig:apps_exp}	
\end{figure}

\begin{figure}[!h]
	\captionsetup{justification=centering}
	\centering
	\includegraphics[width=0.5\textwidth]{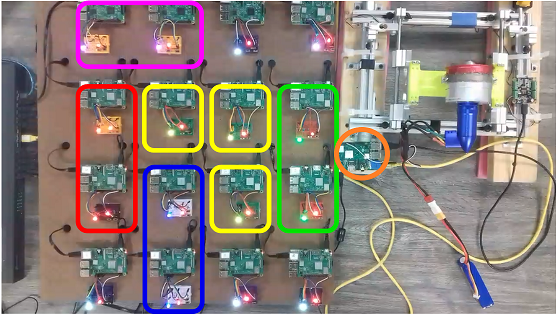}
	\caption{Initial allocation of the applications on the model.\\ 
	The orange circle identifies the extra Raspberry Pi for interactions with the stand.}
	\label{fig:init_alloc}	
\end{figure}

\subsection{Results}\label{sec:results}

Starting  from the initial allocation given in Fig. \ref{fig:result}, faults are triggered on the model. After each fault, the allocators detect the faulty Raspberry Pi and compute a new allocation that is then broadcast on the network. They maintain the execution of the safety-critical application as long as enough resources are available for it. \\

CUs surrounded by faulty neighbors are isolated from the rest of the platform and cannot communicate. As enforced by the communication constraints described in paragraph \ref{par:comm_constr}, such a CU is not given any task to execute and is as good as faulty, as seen in Fig. \ref{fig:result_a}. \\

When a CU recovers from a fault, an application can be allocated back to it as seen in Fig. \ref{fig:result_b}. Applications are dropped according to their priority when more computing units become faulty. However, as illustrated in Fig. \ref{fig:result_d}, when no space is available for all $1^{st}$ priority applications, lower priority ones are still allowed to be executed. \\

The voting system implemented on the fan needs at least two functioning and coherent controllers to run the fan (Fig. \ref{fig:result_d}), as previously explained in section \ref{subsec:soft_comp}: in case the signals received from the controllers are incoherent, it decides not to trust any of them and the engine stops, as in Fig. \ref{fig:result_e}. \\

Since the controllers have the highest priority, they are the last remaining applications to be executed in Fig. \ref{fig:result_g}. After this step, further faults will affect the controllers but no reallocation can happen because no more allocator is executed. \\

\begin{figure*}[h]
	\centering
	\begin{subfigure}[t]{0.5\textwidth}
		\captionsetup{justification=centering}
		\centering
		\includegraphics[height=1.65in]{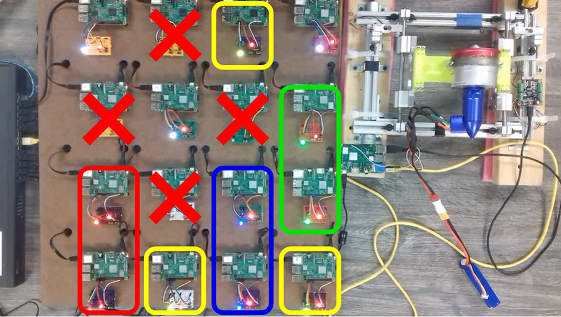}
		\caption{After 4 faults, Application 3 has to be dropped. The isolated Tile cannot be used.}
		\label{fig:result_a}
	\end{subfigure}%
	\begin{subfigure}[t]{0.5\textwidth}
		\captionsetup{justification=centering}
		\centering
		\includegraphics[height=1.65in]{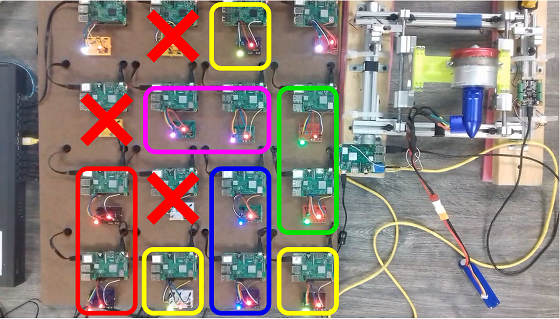}
		\caption{When a Tile recovers from a fault, an allocation can be reallocated to it.}
		\label{fig:result_b}
	\end{subfigure}
	\begin{subfigure}[t]{0.5\textwidth}
		\captionsetup{justification=centering}
		\centering
		\includegraphics[height=1.65in]{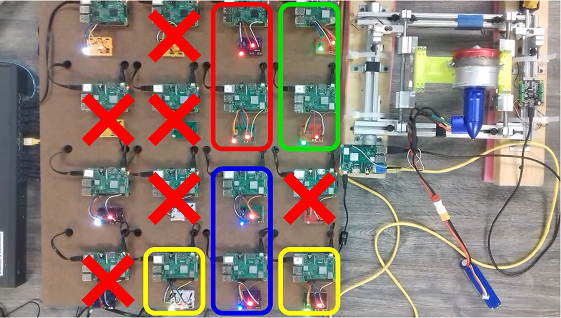}
		\caption{After more faults, Application 3 and a controller have to be dropped by lack of resources.}
		\label{fig:result_c}
	\end{subfigure}%
	\begin{subfigure}[t]{0.5\textwidth}
		\captionsetup{justification=centering}
		\centering
		\includegraphics[height=1.65in]{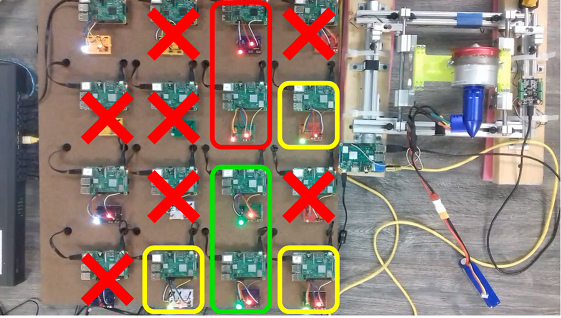}
		\caption{One controller is dropped. The 3 controllers can still be executed, even if all 1st priority applications are not.}
		\label{fig:result_d}
	\end{subfigure}
	\begin{subfigure}[t]{0.5\textwidth}
		\captionsetup{justification=centering}
		\centering
		\includegraphics[height=1.65in]{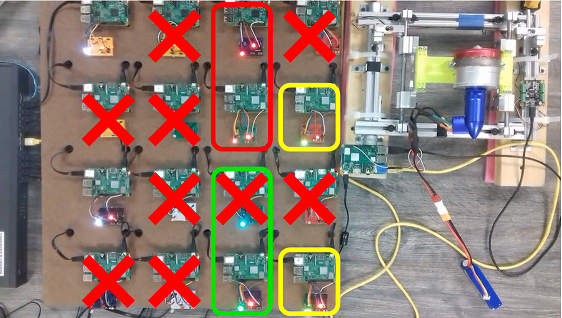}
		\caption{One of the two controllers is affected by a computational fault: the fan stops since the fan voter does not trust any of the two incoherent controllers.}
		\label{fig:result_e}
	\end{subfigure}%
	\begin{subfigure}[t]{0.5\textwidth}
		\captionsetup{justification=centering}
		\centering
		\includegraphics[height=1.65in]{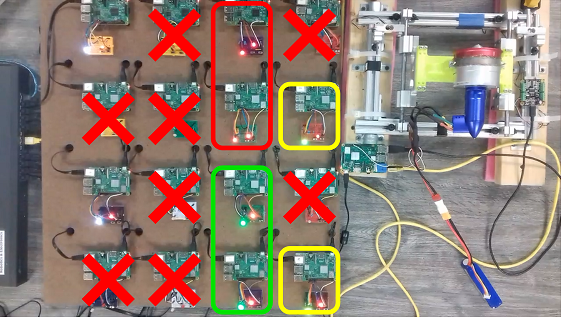}
		\caption{The computational faults disappear: the fan starts again.}
		\label{fig:result_f}
	\end{subfigure}

	\begin{subfigure}[t]{0.5\textwidth}
		\captionsetup{justification=centering}
		\centering
		\includegraphics[height=1.65in]{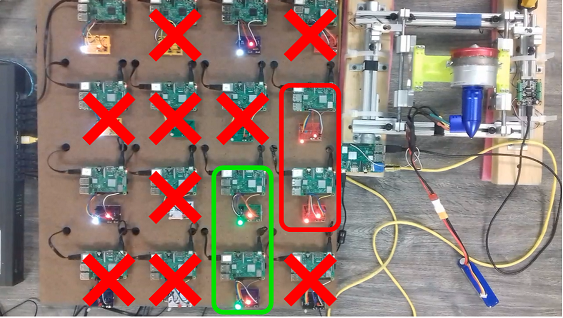}
		\caption{All allocators have been dropped: no further reallocation is possible.}
		\label{fig:result_g}
	\end{subfigure}
	\caption{Result of the task allocation algorithm. A full video of the demo is available at the link below.}
	\label{fig:result}
	\color{blue} \footnotesize \url{https://gtvault-my.sharepoint.com/:v:/g/personal/lsutter6_gatech_edu/EQbz60ttNU1KkzFo0l0tycIBDyboI9KU0SHs4ntq8lPwAw?e=glyKtj}
\end{figure*}

%% file: subfiles/conclusion.tex
This work presented a decentralized allocation algorithm for parallel computing architectures, where individual Computational Units can be affected by faults. The described method consisted in representing the architecture by an abstract graph and formulating the allocation problem as an optimization problem, with the form of a Integer Linear Program. \\

Decentralizing the allocation process has been achieved through redundancy of the allocator executed on the architecture. That way, no centralized element decides of the allocation of the entire architecture. \\

An experimental reproduction of a parallel computing architecture has also been built. It has been used to demonstrate the capabilities of the proposed allocation process to maintain operation of a physical system in a decentralized way while individual component fail. \\


The proposed work assumed that faults affecting the Computational Units of the architecture were automatically detected by the allocation algorithm, so that it is able to compute a new allocation every time a fault affects a Computational Unit. This work can be improved by defining a more precise model of the considered faults and a method to detect them. One first approach to identify dead Computational Unit would be a simple heartbeat that each would send to the allocators. 
A CU not sending its heartbeat would be considered faulty. One challenge to tackle in this approach is the fact that the allocators do not have a fixed position in the architecture, and therefore, the heartbeat of each CU would have to be broadcast through the entire architecture to be sure to reach all allocators. Another solution would be to include the position of the allocator in the allocation message that they broadcast, so that the CU know where to send back their heartbeat. In both cases, the amount communication packets transmitted through the architecture drastically increases. \\

The second lead for improvement is the way communication isolation are taken into account in the allocation problem. For now, only individual nodes with all of their neighbors being faulty were considered isolated. However, an entire area of the architecture, can be isolated from the allocator. The problem becomes quite tricky when the architecture is split in two halves that are isolated one from the other: a decision must be made to decide which area is isolated from the other. It seems that the area with the highest number of allocators should be privileged, since they are the ones that will send the new allocation to other CU, and therefore the ones in the other isolated area will not be able to receive this new allocation. They should therefore be considered as lost CUs.\\

Also, it should be considered that not only the Computational Units can fails, but also the communication links between them. The effect of such faults would be the same as isolating the Computational Units from their neighbors and would make more of them unavailable. It would also change the communication paths usable to connect the allocators to other applications and would affect their position since minimizing these paths is a part of the optimization problem. \\

%% file: subfiles/appendix_v2.tex
\subsection*{Coefficients of the objective function}\label{app:coef}

This appendix provides the proof that the coefficients in the objective function from equation \ref{eq:obj_func} allow to meet the requirements stated in Section \ref{subsec:opt_prob}. For convenience, this objective function is rewritten here:

\begin{equation*}
	\begin{split}
	\max 
	\Bigg\{
	f(\boldsymbol{x}) = \sum_{k = 1}^{N_{\textrm{apps}} } \alpha_k \cdot r_k &- (\beta + 1)\sum_{j = 1}^{N_{\textrm{nodes}}} M_j \\
	&- \sum_{k = 1}^{N_{\textrm{realloc}}}\sum_{j = 1}^{N_{\textrm{CUs}}}\sum_{i = 1}^{N_{\textrm{paths}}} \left|X^{\textrm{Comm, } k}_{ij}\right| 
	\Bigg\} ,
	\end{split}
	\tag{\ref{eq:obj_func}}
\end{equation*}
where 
\begin{equation}
\begin{split}
\beta &= N_{\textrm{realloc}} \times N_{\textrm{CUs}} \times N_{\textrm{paths}} \\
\alpha_{N_{\textrm{apps}}} &= 
(\beta+1) \times N_{\textrm{nodes}} + \beta + 1 \\
\text{and } \forall k < N_{\textrm{apps}}: \\
\alpha_k &= \sum_{l = k+1}^{N_{\textrm{apps}}} \alpha_l + (\beta+1) \times N_{\textrm{nodes}} + \beta + 1 .
\end{split} 
\tag{\ref{eq:obj_func_coeff}}
\end{equation}\\

The requirements of Section \ref{subsec:opt_prob} are also rewritten below.
\begin{enumerate}
	\item When solving the optimization problem, the objective function \ref{eq:obj_func} privileges executing any given application, even if it implies more reallocations and longer communication paths.
	\item When solving the optimization problem, the objective function \ref{eq:obj_func} privileges minimizing the number of reallocations, even if it implies longer communication paths.
	\item When solving the optimization problem, the objective function \ref{eq:obj_func} privileges executing a given application compared to running any number of applications with a lower priority. \\
\end{enumerate}


The following theorems prove that these requirements are met.\\

\begin{theorem} \label{thm:1}
	$\forall \ \tilde{k} \in \llbracket 1,\ N_{\textrm{apps}} \rrbracket:$ 
	$$\alpha_{\tilde{k}} > (\beta + 1)\sum_{j = 1}^{N_{\textrm{nodes}}} 1 + \sum_{k = 1}^{N_{\textrm{realloc}}}\sum_{j = 1}^{N_{\textrm{CUs}}}\sum_{i = 1}^{N_{\textrm{paths}}} 1, $$
	that is, the contribution to the value of the objective function for executing application $\mathrm{app}_{\tilde{k}}$ is greater than the maximum contribution for reducing the number of reallocations and the length of the communication paths. 
\end{theorem}

\begin{proof}
	$\forall \ \tilde{k} \in \llbracket 1,\ N_{\textrm{apps}} \rrbracket:$
	$$\alpha_{\tilde{k}} \geq (\beta+1) \times N_{\textrm{nodes}} + \beta,$$ by definition of $\alpha_{\tilde{k}}$.
	
	Now, $$(\beta + 1)\sum_{j = 1}^{N_{\textrm{nodes}}} 1 + \sum_{k = 1}^{N_{\textrm{realloc}}}\sum_{j = 1}^{N_{\textrm{CUs}}}\sum_{i = 1}^{N_{\textrm{paths}}} 1 = (\beta+1) \times N_{\textrm{nodes}} + \beta .$$
	
	So $$\alpha_{\tilde{k}} > (\beta + 1)\sum_{j = 1}^{N_{\textrm{nodes}}} 1 + \sum_{k = 1}^{N_{\textrm{realloc}}}\sum_{j = 1}^{N_{\textrm{CUs}}}\sum_{i = 1}^{N_{\textrm{paths}}} 1. $$
\end{proof}

Theorem \ref{thm:1} proves that requirement 1 is met. \\

\begin{theorem} \label{thm:2}
	$$(\beta + 1) \times 1 > \sum_{k = 1}^{N_{\textrm{realloc}}}\sum_{j = 1}^{N_{\textrm{CUs}}}\sum_{i = 1}^{N_{\textrm{paths}}} 1, $$
	that is, the contribution to the value of the objective function for not reallocating one Application node is greater then the maximum contribution for reducing the length of communication paths.
\end{theorem}

\begin{proof}
	
	$$ \beta + 1 > \beta = N_{\textrm{realloc}} \times N_{\textrm{CUs}} \times N_{\textrm{paths}} = \sum_{k = 1}^{N_{\textrm{realloc}}}\sum_{j = 1}^{N_{\textrm{CUs}}}\sum_{i = 1}^{N_{\textrm{paths}}} 1 $$.
	
\end{proof}

Theorem \ref{thm:2} proves that requirement 2 is met. \\

\begin{theorem} \label{thm:3}
	$\forall \ \tilde{k} \in \llbracket 1,\ N_{\textrm{apps}} - 1 \rrbracket:$ 
	$$\alpha_{\tilde{k}} > \sum_{l = \tilde{k}+1}^{N_{\textrm{apps}}} \alpha_l $$
	that is, the contribution to the value of the objective function for executing application $\mathrm{app}_{\tilde{k}}$ is greater than the contribution for executing every applications with lower priority than $\mathrm{app}_{\tilde{k}}$, which are $\mathrm{app}_{\tilde{k}+1}$ to $\mathrm{app}_{N_{\textrm{apps}}}$.
\end{theorem}

\begin{proof}
	
	$\forall \ \tilde{k} \in \llbracket 1,\ N_{\textrm{apps}} \rrbracket:$
	$$\alpha_{\tilde{k}} = \sum_{l = \tilde{k}+1}^{N_{\textrm{apps}}} \alpha_l + (\beta+1) \times N_{\textrm{nodes}} + \beta + 1 > \sum_{l = \tilde{k}+1}^{N_{\textrm{apps}}} \alpha_l $$
	since $ (\beta+1) \times N_{\textrm{nodes}} + \beta + 1 > 0 $.
\end{proof}

Theorem \ref{thm:3} proves that requirement 3 is met. \\